\documentclass[12pt]{article}%

\usepackage{graphicx}
\usepackage{amsthm}
\usepackage{amsmath}
\usepackage{amsfonts}
\usepackage{amssymb}

\def\bfK{\mathbf K}
\def\bfP{\mathbf P}

\def\bft{\mathbf t}
\def\ubft{\underline\bft}
\def\bfsigma{{\mbox{\boldmath${\sigma}$}}}
\def\bfpsi{{\mbox{\boldmath${\psi}$}}}
\def\bfphi{{\mbox{\boldmath${\phi}$}}}
\def\bfpi{{\mbox{\boldmath${\pi}$}}}
\def\ubfsigma{\underline{\mbox{\boldmath${\bfsigma}$}}}
\def\bfLambda{{\mbox{\boldmath${\Lambda}$}}}
\def\bfPi{{\mbox{\boldmath${\Pi}$}}}

\def\wt{\widetilde}
\def\wtU{\widetilde U}
\def\wtu{\widetilde u}
\def\diy{\displaystyle}

\newtheorem{theo}{Theorem}[section]
\newtheorem{lm}{Lemma}[section]
\newtheorem{df}{Definition}[section]

\newtheorem{prop}{Proposition}[section]

\newtheorem{remark}{Remark}

\numberwithin{equation}{section}

\begin{document}

\begin{center}

{ \large \bf Bounds on the critical line via transfer matrix methods for an Ising model coupled to causal dynamical triangulations}

\vspace{30pt}

{\sl J.C.~ Hernandez}$\,^{a}$ , {\sl Y.~Suhov}$\,^{a,b,c}$, {\sl A.~Yambartsev}$\,^{a}$ and
{\sl S.~Zohren}$\,^{d,e,a}$

\vspace{24pt}

{\footnotesize

$^a$~Department of Statistics, Institute of Mathematics and Statistics, \\ University of S\~ao Paulo, Rua do Mat\~ao, 1010, S\~ao Paulo, CEP 05508-090, Brazil

\vspace{10pt}

$^b$~DPMMS, University of Cambridge, \\ Wilberforce Road, Cambridge CB3 0WB, UK

\vspace{10pt}

$^c$ IITP RAS, 19 Bol'shoi Karetnyi per., Moscow, 127994 Russia

\vspace{10pt}

$^d$~Physics Department, PUC Rio de Janeiro \\Rua Marqu\^es de S\~ao Vincente 225, G\'avea, Rio de Janeiro, Brazil
\vspace{10pt}

$^e$~Rudolf Peierls Center for Theoretical Physics \& Mansfield College, \\ University of Oxford, 1 Keble Road, OX1 3NP Oxford, UK
}

\vspace{48pt}

\end{center}

\begin{abstract}
We introduce a transfer matrix formalism for the (annealed) Ising model coupled to two-dimensional causal dynamical triangulations. Using the Krein-Rutman theory of positivity preserving operators we study several properties of the emerging transfer matrix. In particular, we determine regions in the quadrant of parameters $\beta ,\mu >0$ where the infinite-volume free energy converges, yielding results on the convergence and asymptotic properties of the
partition function and the Gibbs measure.
\end{abstract}


\newpage

\section{Introduction. A review of related results}

In the study of two-dimensional quantum gravity and non-critical string theory, models of
discrete random surfaces play an essential role.

In the 1980s, so-called \emph{dynamical triangulations} (DT) were introduced
to define a Euclidean path integral for two-dimensional quantum gravity (see
\cite{Ambjorn:1997di} for an overview). In particular, the partition function has been
determined as a sum over all possible triangulations of a sphere where each configuration
is weighted by a Boltzmann factor $e^{-\mu |T|}$, with $|T|$ standing for the size of
the triangulation and $\mu$ being the cosmological constant. The evaluation of the
partition function was reduced to a purely combinatorial problem that can be solved
with the help of the early work of Tutte \cite{Tutte1962a,Tutte1963}; alternatively,
more powerful techniques were proposed, based on random matrix models (see, e.g.,
\cite{DiFrancesco:1993nw}) and bijections to well-labelled trees (see \cite{Schaeffer1997,
bouttier-2002-645}). One can then pass to a continuum limit by taking the number of
triangles to infinity. An interesting property of the resulting ``quantum geometry''
is its fractal structure as illustrated in Figure \ref{fig1} (a). In the physical literature
such fractal structures are called ``baby universes'',
and they completely dominate the continuum limit leading to a fractal dimension $d=4$, 
where the fractal dimension is defined through the behaviour of the number of triangles or the number of vertices within a given graph distance $R$ from a chosen vertex, $B(R)$, through $B(R)\sim R^d$ as $R\to\infty$.

From a probabilistic point of view there has recently been an increasing interest in DT, most notably
through the work of Angel and Schramm on a uniform measure on infinite planar triangulations
\cite{Angel:2002ta}, as well as through the work of Le Gall, Miermont and collaborators on
Brownian maps (see \cite{Clay-Le-Gall} for a recent review).

From a physical point of view it is interesting to study various models of matter, such as the
Ising model, coupled to the DT. The calculation of the partition function in this case also
reduces to a combinatorial problem. It was first solved in 
\cite{Kazakov:1986hu,Boulatov:1986sb} by using random matrix models and later by using
a bijection to well-labelled trees \cite{Schaeffer-Ising}. It is interesting that the
solution here is much simpler than in the case of a flat triangular or square lattice as
given by Onsager \cite{Onsager}. Further, one can see that the critical exponents in the case where the model
is coupled to DT differ
from the Onsager values. This is related to the strong back-reaction of the Ising model
with the quantum geometry. In particular, the spin clusters energetically prefer to sit
within baby universes since those are connected to the main universe through a very short
so-called bottleneck boundary (see Figure \ref{fig1} (b)). The spins increase the
fractal structure leading to a change in the values of the critical exponents at the
critical temperature.

\begin{figure}[t]
\begin{center}
\includegraphics[width=14cm]{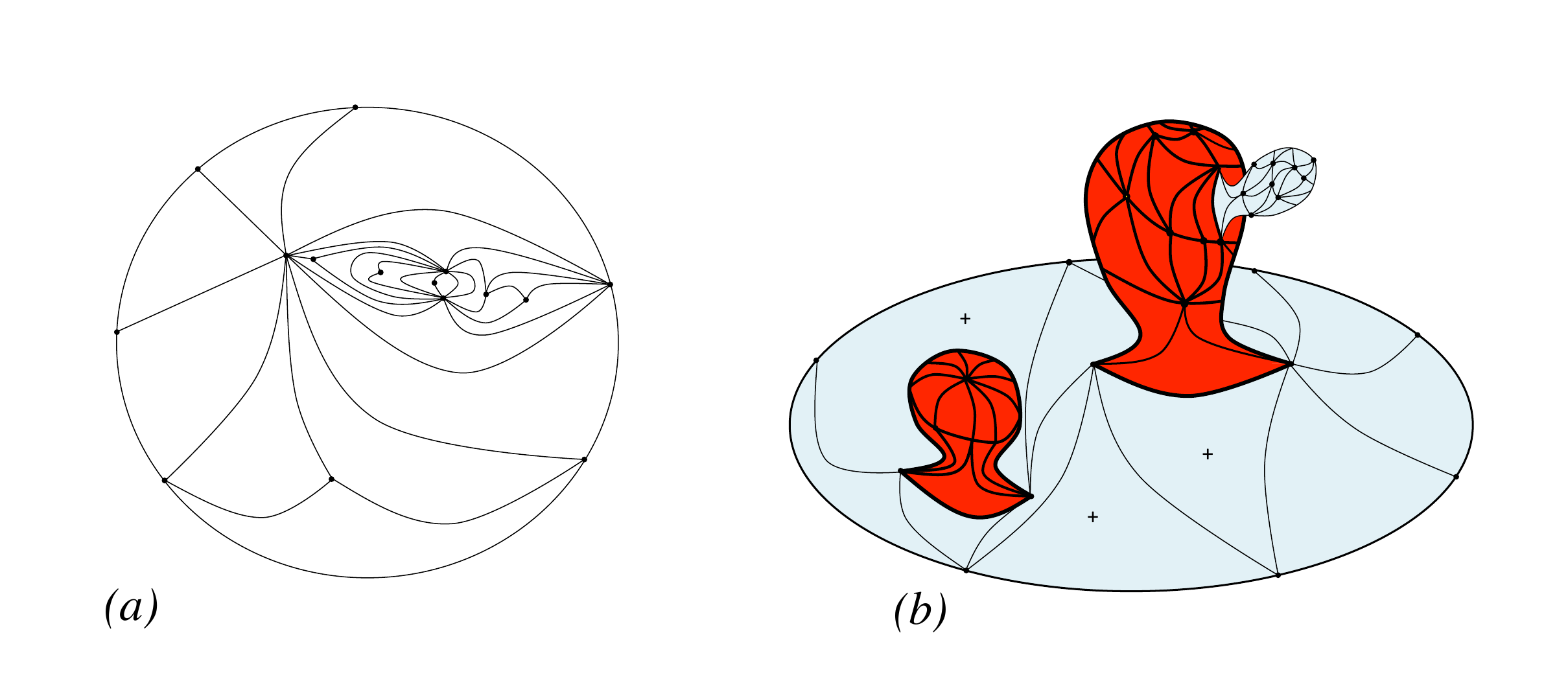}
\end{center}
\caption{(a) A section of a typical planar triangulation of the DT ensemble illustrating
the fractal structure of the quantum geometry. (b) Illustration of the baby
universes and
the formation of spin clusters within them. Each baby universe corresponds
to a fractal structure as in part (a) drawn out of the plane.}%
\label{fig1}
\end{figure}

In a continuum framework one attempts to understand the resulting theory as a Liouville
theory coupled to a conformal field theory with central charge $c=1/2$. Furthermore, this
leads to a simple algebraic identity (the KPZ-relation) between the critical exponents
of the Ising model on a flat lattice and the critical exponents of the Ising model coupled
to DT \cite{Knizhnik:1988ak}.

While DT has a very rich mathematical structure which very recently has been related to the SLE
(the Schramm-Loewner evolution) and level curves of a Gaussian free field \cite{Duplantier:2010yw},
from the point of view of quantum gravity its fractal structure leads to causality-violating
geometries that are arguably non-physical. This led to the development of so-called
\emph{causal dynamical triangulations} (CDT) by Ambj{\o}rn and Loll \cite{Ambjorn:1998xu},
to define the \emph{Lorentzian} gravitational path integral. A causal triangulation is
formed by triangulations of spatial strips as illustrated in Figure \ref{fig2}. Note that
the left and right boundaries of the spatial strip are periodically identified.

\begin{figure}[t]
\begin{center}
\includegraphics[width=12cm]{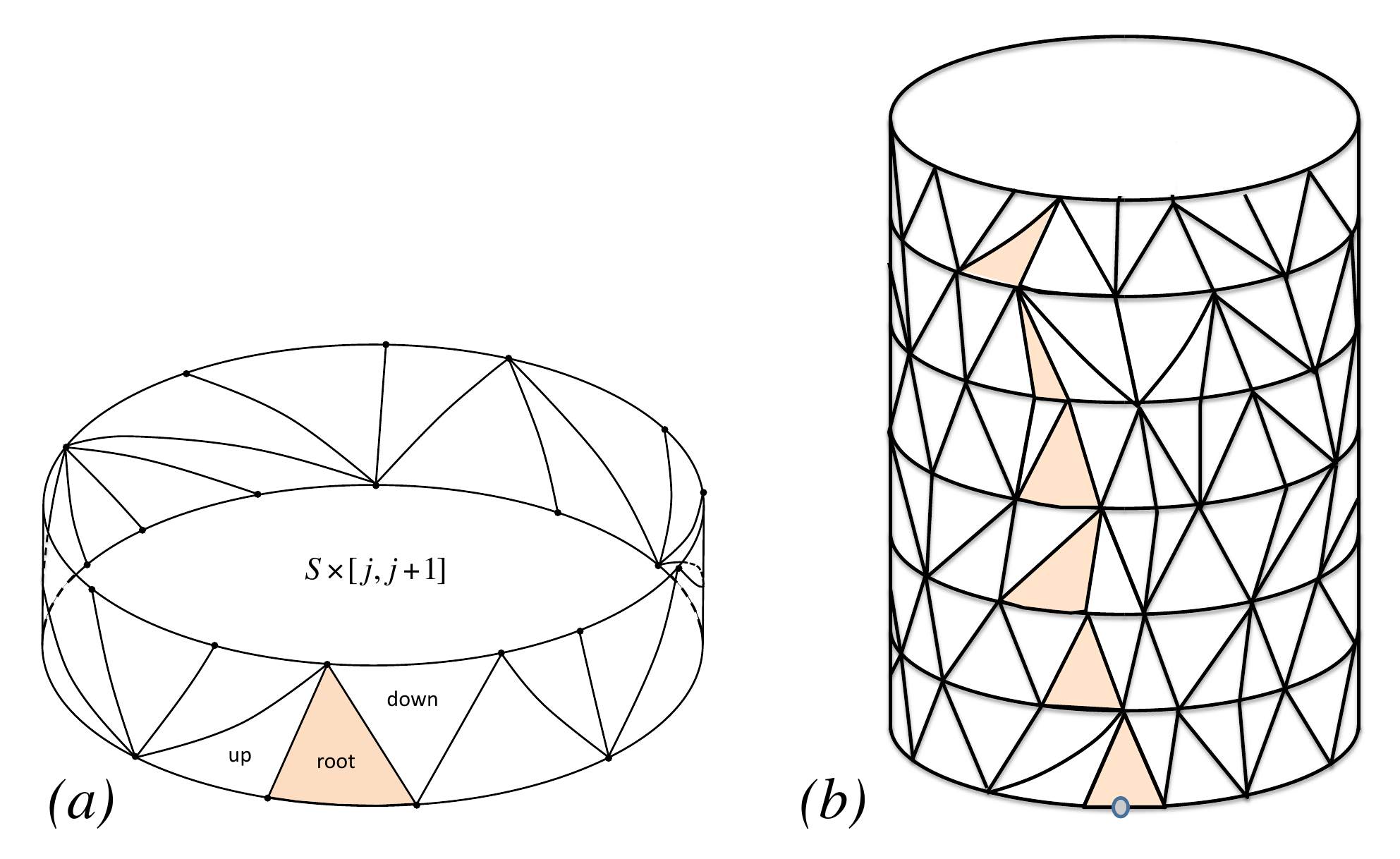}
\end{center}
\caption{(a) A strip of a causal triangulation of ${\mathcal S}\times [j,j+1]$. (b) A a causal triangulation of ${\mathcal S}\times [j,j+6]$ composed of six strips. The rooted edge on the lower boundary induces a sequence of root triangles in each strip.}%
\label{fig2}
\end{figure}

A first analytical solution of two-dimensional (pure) CDT was obtained in
\cite{Ambjorn:1998xu} where it was shown that the resulting quantum geometry, while still
random, is much more regular than in the case of a DT, leading to a fractal dimension of
$d=2$. From a probabilistic point of view, we would like to note that a uniform measure
on infinite causal triangulations UICT has been recently introduced in
\cite{Durhuus:2009sm,SYZ1}.

An interesting question is: What are the properties of the Ising model coupled to a CDT ensemble? As
was said above, the CDT ensemble is more regular than that of the DT, but it is still random
and allows for a back-reaction of the spin system with the quantum geometry. Monte Carlo
simulations \cite{Ambjorn:1999gi} (see also \cite{Benedetti:2006rv,Ambjorn:2008jg}) give
a strong evidence that critical exponents of the Ising model coupled to CDT are identical
to the Onsager values.

While recently much progress has been made in the development of analytical techniques for CDT \cite{Ambjorn:2007jm,Ambjorn:2008ta}, particularly random matrix models
\cite{Ambjorn:2008jf,CDTmatrix2,Ambjorn:2008gk}, and their application to multi-critical CDT \cite{multi,dimer,multi2}, 
the causality constraints still makes it
difficult to find an analytical solution of the Ising model coupled to CDT. For the
quenched Ising model coupled to two-dimensional CDT some progress has been made in proving
the existence of a phase transition \cite{anatoli}.

In this article we develop a transfer matrix formalism for the annealed Ising model coupled
to two-dimensional CDT. Spectral properties of the transfer matrix are rigorously analysed
by using the Krein-Rutman theorem \cite{Krein-Rutman} on operators preserving the cone
of positive functions. This yields results on convergence and asymptotic properties of the
partition function and the Gibbs measure and allows us to determine regions in the 
parameter quarter-plane where the partition function converges. 

\emph{Outline:} The article is organised as follows. Section \ref{sec2} contains basic definitions. Section  \ref{sec3} gives
a summary of the transfer-matrix formalism for CDTs. In Section  \ref{sec4} we introduce the
transfer matrix for the Ising model coupled with a CDT and state and comment on
our main results. This is followed with concluding remarks in Section \ref{sec5}. An Appendix contains 
the proofs of the stated results.  

\section{Definitions}\label{sec2}

We will work with rooted causal dynamic triangulations of the cylinder 
$C_N = {\mathcal S}\times [0,N]$, $N = 1, 2, \dots$, which 
have $N$ bonds (strips) ${\mathcal S}\times [j,j+1]$.  Here ${\mathcal S}$ stands for
a unit circle.  The definition of a causal triangulation starts by considering  
a connected graph $G$ embedded in $C_N$ with the property that all faces of $G$
are triangles (using the convention that an edge incident to the same face on two
sides counts twice, see \cite{SYZ1} for more details). A triangulation $\ubft$ of $C_N$ is 
a pair formed by a graph $G$ with the above property and the set $F$ of all
its (triangular) faces: $\ubft = (G, F)$.

\begin{df} \label{defact}
A triangulation $\ubft$ of $C_N$ is called a {\it causal triangulation} (CT) if the 
following conditions hold:
\begin{itemize}
\item each triangular face of $\ubft$ belongs to some strip $\mathcal S \times [j, j + 1]$, $j =
1, \dots, N-1$, and has all vertices and exactly one edge on the boundary
$(\mathcal S \times \{j\}) \cup (\mathcal S\times \{j+1\})$ of the strip $\mathcal S\times [j, j + 1]$;

\item if $k_j = k_j(\ubft)$ is the number of edges on $\mathcal S \times \{j\}$, then we have
$0 < k_j < \infty$ for all $j = 0, 1, \dots, N-1$.
\end{itemize}
\end{df}

\begin{df}\label{defroot}
A triangulation $\ubft$ of $C_N$ is called rooted if it has a root. The root in the 
triangulation $\;\ubft\;$ is represented by a triangular face $t$ of $\;\ubft$, called the root
triangle, with an anticlock-wise  ordering on its vertices $(x, y, z)$ where
$x$ and $y$  belong to $S^1\times \{0\}$. The vertex $x$ is identified as the root
vertex and the (directed) edge from $x$ to $y$ as the root edge. 
\end{df}

\begin{df}
Two causal rooted triangulations of $C_N$, say $\;\ubft = (G, F)$ and
$\ubft^\prime = (G^\prime, F^\prime)$, are equivalent if there exists a self-homeomorphism of $C_N$ 
which (i) transforms each slice $S^1\times \{j\}$, $j = 0,\dots, N-1$ to itself and preserves
its direction, (ii) induces an isomorphism of the graphs $G$ and $G^\prime$ and a bijection
between $F$ and $F^\prime$, and (iii) takes the root of $\ubft$ to the root of $\ubft^\prime$.
\end{df}

A triangulation $\ubft$  of $C_N$ is identified as a consistent sequence
$$\ubft = (\bft(0), \bft(1), \dots, \bft (N-1)),$$ 
where $\bft (i)$ is a causal triangulation of the 
strip ${\mathcal S}\times [i,i+1]$ (see Figure \ref{fig2}). The latter means that each $\bft (i)$ is described by a partition of
${\mathcal S}\times [i,i+1]$ into triangles where each triangle has one vertex on 
one of the slices ${\mathcal S}\times\{i\}$, ${\mathcal S}\times\{i+1\}$ and two
on the other, together with the edge joining these two vertices. The property of 
consistency means that each pair $(\bft(i), \bft(i+1))$ is consistent, i.e., every side of a triangle from 
$\bft(i)$ lying in ${\mathcal S}\times\{i+1\}$ serves as a side of a triangle
from $\bft(i+1)$, and vice versa.  

The triangles forming the causal triangulation
$\bft (i)$ are denoted by $t(i,j)$, $1\leq j\leq n(\bft (i))$ where, $n(\bft (i))$ stands
for the number of triangles in the triangulation $\bft(i)$. The enumeration of these triangles starts with 
what we call the root triangle in $\bft(i)$; it is determined recursively as follows (see Figure \ref{fig2} (b)): 
First, we have the root triangle $t(0,1)$ in $\bft (0)$ (see Definition 2.2). Take the  vertex 
of the triangle $t(0,1)$ which lies on the slice ${\mathcal S}\times \{1\}$ and denote it by $x'$. 
This vertex is declared the root vertex 
for $\bft (1)$. Next, the root edge for $\bft (1)$ is the one incident to $x'$ and lying on 
${\mathcal S}\times \{1\}$, so that if $y'$ is its other end and $z'$ is the third 
vertex of the corresponding triangle then $x',y',z'$ lists the three vertices anticlock-wise.
Accordingly, the triangle with the vertices $x',y',z'$ is called the root triangle for $\bft (1)$.
This construction can be iterated, determining  the root vertices, root edges and root triangles
for $\bft(i)$, $0\leq i\leq N-1$.

It is convenient to introduce the notion of  ``up" and ``down" triangles (see Figure~\ref{fig2} (a)). 
We call a triangle $t \in \bft (i)$ an up-triangle if it has an edge on the slice $\mathcal S \times \{ i\}$   
and a down-triangle if it has an edge on the slice $\mathcal S \times \{ i+1\}$. By Definition 2.1,
every triangle is either of type up or down. Let $n_{\rm{up}}(\bft (i))$ and $n_{\rm{do}}(\bft (i))$ 
stand for the number of up- and down-triangles in the triangulation $\bft (i)$. 

Note that for any edge lying on the slice $\mathcal S \times \{ i\}$ belongs to exactly two 
triangles: one up-triangle from $\bft (i)$ and one down-triangle from $\bft (i-1)$. This provides 
the following relation: the number of triangles in the triangulation $\ubft$ is twice the total number of 
edges on the slices. More precisely, let $n^i$ be the number of edges on slice $\mathcal S \times \{i\}$. 
Then, for any $i=0,1,\dots, N-1$,
\begin{equation}\label{et2-yamb}
n(\bft(i)) = n_{up} (\bft (i)) + n_{do}(\bft (i)) = n^i + n^{i+1},
\end{equation}
implying that 
\begin{equation}\label{et1-yamb}
\sum_{i=0}^{N-1} n(\bft(i)) = 2 \sum_{i=0}^{N-1} n^i.
\end{equation}

There is another useful property regarding the counting of triangulations. Let us fix the number of edges $n^i$ and $n^{i+1}$ in the slices $\mathcal S \times \{i\}$ and $\mathcal S \times \{i+1\}$. The number of possible rooted CTs
 of the slice $\mathcal S \times [i,i+1]$ with $n^i$ up- and $n^{i+1}$ down-triangles is equal to
\begin{equation}\label{noftr-yamb}
\binom{ n^i + n^{i+1} -1 }{n^i -1} = \binom{ n(\bft(i)) -1 }{n_{up}(\bft(i)) -1}\,.
\end{equation}

\section{Transfer matrix formalism for pure CDTs}\label{sec3}


We begin by discussing the case of pure causal dynamical triangulations,
as was first introduced in \cite{Ambjorn:1998xu} (see also \cite{MYZ2001}
for a mathematically more rigorous account). 

The partition function for rooted CTs in the cylinder $C_N$ with periodical
spatial boundary conditions (where $\bft (0)$ is consistent with $\bft (N-1)$) and
for the value of the cosmological constant $\mu$ is given by
\begin{equation} \label{yamb-pf1}
Z_N(\mu)=\sum_{\ubft} e^{-\mu n(\ubft) } = \sum_{(\bft(0), \dots, \ubft(N-1))} \exp \Bigl\{-\mu \sum_{i=0}^{N-1} n(\bft(i)) \Bigr\}.
\end{equation}
Using the properties (\ref{et1-yamb}) and (\ref{noftr-yamb}) we can represent the partition function (\ref{yamb-pf1}) in the following way
\begin{equation} \label{yamb-pf1-1}
Z_N(\mu) =
\sum_{n^0\geq 1,\dots,n^{N-1} \geq 1} \exp\Bigl\{-2\mu\sum_{i=0}^{N-1} n^i \Bigr\}\prod_{i=0}^{N-1} \binom{n^i+n^{i+1}-1}{n^{i}-1}. 
\end{equation}
Moreover, $Z_N(\mu)$ admits a trace-related representation 
\begin{equation}\label{yamb-pf1-2}
Z_N(\mu)= \mbox{tr}\; \bigl( U^N \bigr).
\end{equation}
This gives rise to a transfer matrix $U=\{u(n,n^\prime)\}_{n,n^\prime = 1, 2, \dots}$
describing the transition from one spatial strip to the next one. It is an infinite matrix with
strictly positive entries
\begin{equation}\label{yamb-tmpg}
    u(n,n^\prime) = \binom{n+n^\prime-1}{n-1} g^{n+n^\prime}.
\end{equation}
For notational convenience we use the parameter $g=e^{-\mu}$ (a single-triangle fugacity). The
entry $u(n,n^\prime)$ yields the number of possible triangluations of a single strip 
(say, ${\mathcal S}\times [0,1]$) with $n$ lower boundary edges 
(on ${\mathcal S}\times\{0\}$) and $n^\prime$ upper boundary edges 
(on ${\mathcal S}\times\{1\}$). See Figure \ref{fig2}. The asymmetry
in $n$ and $n^\prime$ is due to the fact that the lower boundary is marked while the upper one
is not. However, a symmetric transfer matrix ${\wtU} =\{{\wtu}(n,n')\}$ can be
introduced,  associated with a
strip where both boundaries are kept unmarked:
\begin{equation}
{\wtu} (n,n^\prime) = n^{-1} u(n,n^\prime).
\end{equation}

The $N$-strip Gibbs distribution ${\mathcal P}_N$ assigns the following
probabilities to strings $(n^0,\dots ,n^{N-1})$ with the number of
triangles $n^i \ge 1$ for all $i=0,\dots, N-1$:
\begin{equation}\label{GibbsCDTN}
{\mathcal P}_N(n^0,\dots ,n^{N-1} ) = \frac{1}{Z_N(\mu)}
\exp \Bigl\{-2\mu\sum_{i=0}^{N-1} n^i \Bigr\}\prod_{i=0}^{N-1} \binom{n^i+n^{i+1}-1}{n^i-1}.
\end{equation}

We state two lemmas featuring properties of matrix $U$:
\begin{lm}\label{yamb-l1}
For any $g>0$ the matrix
$U$ and its transpose $U^{\rm T}$ have an eigenvalue $\Lambda=\Lambda (g)$ given by
\begin{equation}\label{Lambda(g)}
\Lambda(g)=\left[(1- \sqrt{1-4g^2})/(2g)\right]^2.
\end{equation}
The corresponding eigenvectors
$$\phi=\{\phi(n)\}_{n=1,2,...}\;\hbox{
and }\;\phi^*=\{\phi^*(n)\}_{n=1,2,...}$$
have entries
\begin{equation}\label{phi(n)}
\phi(n)= n \big(\Lambda (g)\big)^n,\;\;\phi^*(n)=(\Lambda(g))^n.
\end{equation}
\end{lm}
\begin{proof}
A direct verification shows that
\[
\sum_{n^\prime} u(n,n^\prime) n^\prime \Lambda^{n^\prime}(g) = n\Lambda^{n+1}(g)
\;\hbox{ and }\;\sum_n\Lambda^n(g)u(n,n')=\Lambda^{n'+1}(g).
\]
(In fact, each of these relations implies the other.) See Theorem 1 in  \cite{MYZ2001}.
\end{proof}

\begin{lm}\label{sz-l2}
For any fixed $n$ and any $g<1$ (equivalently, $\mu>0$) one has
\begin{equation}\label{yamb-sum1}
    \sum_{n^\prime} u(n,n^\prime) = \Bigl( \frac{g}{1-g} \Bigr)^n \bigl( 1- (1-g)^n \bigr).
\end{equation}
\end{lm}
\begin{proof}
The proof again follows from a straightforward verification.
\end{proof}

A transfer-matrix formalism of Statistical Mechanics predicts that, as $N\to \infty$, the partition
function is governed by the largest eigenvalue $\Lambda$ of the transfer matrix:
\begin{equation}
Z_N(g) = {\rm tr}\;U^N \sim\Lambda^N
\end{equation}
We make this statement more precise in the statements of Lemma \ref{sz-l3} and Theorem
\ref{theo1} below. Here the
symbol $\ell^2$ stands for the Hilbert space of square-summable complex
sequences (infinite-dimensional vectors) $\psi =\{\psi (n)\}_{n=1,2,\ldots}$ equipped
with the standard scalar
product $\langle \psi', \psi'' \rangle=\sum_n \psi'(n) {\overline\psi''}(n)$. Accordingly,
the matrices $U$ and $U^{\rm T}$ are treated as operators on $\ell^2$.
\begin{lm}\label{sz-l3}
For any $g<1/2$ (equivalently $\mu>\ln 2$) the following statements hold true:
\begin{enumerate}
\item $U$ and $U^{\rm T}$ are bounded operators in $\ell^2$ preserving the cone of positive
vectors;

\item The sum $\sum_{n, n^\prime} u(n, n^\prime) < \infty$. Consequently,
$U$ and $U^{\rm T}$ have
$${\rm tr}\;\big(UU^{\rm T}\big)={\rm tr}\;\big(U^{\rm T} U\big)<\infty ,$$
i.e., $U$ and  $U^{\rm T}$ are Hilbert-Schmidt operators. Therefore, $\forall$ $N\geq 2$, 
$U^N$ and $\left(U^{\rm T}\right)^N$ are trace-class operators.

\item The maximal eigenvalue $\Lambda =\Lambda (g)$ of $U$ in $\ell^2$ is positive,
coincides with the maximal eigenvalue of $U^{\rm T}$ and is given by Eqn (\ref{Lambda(g)}).
The corresponding eigenvectors $\phi,\phi^*\in\ell^2$ are unique up to multiplication by
a constant factor and given in Eqn (\ref{phi(n)}).

\item The following asymptotical formulas hold as $N\to\infty$:
\[
\frac{1}{\Lambda^N}\,{\rm tr}\;\bigl(U^N\bigr),\;\; \frac{1}{\Lambda^N}\,{\rm tr}\;\bigr( (U^{\rm T}) ^{N} \bigl) \to 1,
\]
and, $\forall$ vectors $\psi',\psi''\in \ell^2$,
\[\frac{1}{\Lambda^N}\langle\psi',U^N\psi''\rangle\to\langle\psi',\phi\rangle\langle
\phi^*,\psi''\rangle,\]
where the eigenvectors $\phi$ and $\phi^*$ are normalized so that
$\langle\phi ,\phi^*\rangle=1$.
\end{enumerate}
\end{lm}

\begin{theo}\label{theo1}
For any $g<1/2$ the following relation holds true:
\begin{equation}\label{yamb-e13}
\lim_{N\to\infty}\frac{1}{N}\log\,Z_N(g)=\log\,\Lambda
\end{equation}
with $\Lambda=\Lambda(g)$ given in (\ref{Lambda(g)}).
Further, the $N$-strip Gibbs measure ${\mathcal P}_N$ converges weakly to a limiting measure
${\mathcal P}$ which is represented by a positive recurrent Markov chain on
${\mathbb Z}_+=\{1,2,\ldots\}$, with the transition matrix $P=\{P(n,n')\}_{n=1,2,\ldots}$ and
the invariant distribution $\pi$. Here
$$P(n,n')=\frac{u(n,n')\phi (n')}{\Lambda\phi(n)}$$
and
$$\pi(n )=\frac{\phi^*(n)\phi (n)}{\langle\phi^* ,\phi\rangle}.$$
where $\phi(n)$ and $\phi^*(n)$ are as in (\ref{phi(n)}).
\end{theo}

\begin{proof}
The proof is a consequence of Lemma \ref{yamb-l1} and \ref{sz-l3} and the Krein-Rutman theory
\cite{Krein-Rutman}.
\end{proof}

\section{An Ising model coupled to the CDT: statement of results} \label{sec4}

\subsection{The model}

With any triangle from a triangulation $\ubft$ we associate a spin taking values
$\pm 1$. An $N$-strip
configuration of spins is represented by a collection
$$\ubfsigma =(\bfsigma (0), \bfsigma (1), \dots, \bfsigma (N-1))$$ where
$\bfsigma (i)=\bfsigma (\bft (i))$ is a configuration of spins $\sigma (i,j)$
over triangles
$t(i,j)$ forming a triangulation $\bft (i)$, $1\leq j\leq n(\bft(i))$. We will say
that a single-strip configuration of spins $\bfsigma (i)$ is supported by
a triangulation $\bft (i)$ of strip ${\mathcal S}\times [i,i+1]$.
We consider a usual (ferromagnetic) Ising-type energy where two spins
$\sigma (i,j)$ and
$\sigma (i',j')$ interact if their supporting triangles $t(i,j)$, $t(i',j')$
share a common edge; such triangles are called nearest neighbors, and this
property is reflected in the notation $\langle \sigma (i,j),\sigma (i',j')\rangle$, where
we require $0\leq i \le i' \leq N-1$. Thus, in our model each spin has three neighbors.
Moreover, a pair $\langle \sigma (i,j),\sigma (i',j')\rangle$ can only occur for
$i'-i\leq 1$ or $i=0$, $i'=N-1$. Formally, the Hamiltonian of the model reads:
\begin{equation}\label{IsingHami}
{\mathbb H}(\ubfsigma )=-\sum_{\langle \sigma (i,j),\sigma (i',j') \rangle} \sigma (i,j)\sigma (i',j').
\end{equation}
We will use the following decomposition:
\begin{equation}
{\mathbb H}(\ubfsigma)= \sum_{i=0}^{N-1} H(\bfsigma(i)) + \sum_{i=0}^{N-1}V(\bfsigma(i),
\bfsigma(i+1)),
\end{equation}
where we assume that $\bfsigma (0) \equiv\bfsigma (N)$ (the periodic
spatial boundary condition). Here $H(\bfsigma(i))$ represents
the energy of the configuration $\bfsigma (i)$:
\begin{equation}\label{sstripenerg}
H(\bfsigma (i))=-\sum_{\langle \sigma (i,j),\sigma (i,j') \rangle}\sigma (i,j)\sigma (i,j').
\end{equation}
Further, $V(\bfsigma (i),\bfsigma (i+1))$ is the energy of interaction between neighboring triangles
belonging to the adjacent strips ${\mathcal S}\times [i, i+1]$ and ${\mathcal S}\times [i+1, i+2]$:
\begin{equation}\label{bstripenerg}
V(\bfsigma(i),\bfsigma(i+1))=-\sum_{\langle \sigma (i,j),\sigma (i+1,j')\rangle}\sigma (i,j)\sigma (i+1,j').
\end{equation}

The partition function for the (annealed) $N$-strip Ising model coupled to
CDT, at the
inverse temperature $\beta >0$ and for the cosmological constant $\mu$, is given
by
\begin{eqnarray}\label{yamb-pf}
&& \Xi_N(\mu,\beta)=\sum_{(\bft (0),\dots,\bft (N-1))}
\exp\Bigl\{ -\mu \sum_{i=0}^{N-1} n(\bft (i)) \Bigr\} \\ \nonumber &&
\quad\times \sum_{( \bfsigma (0),\dots ,\bfsigma (N-1) ) }
\prod_{i=0}^{N-1}
\exp\,\Bigl\{ -\beta H(\bfsigma (i))-\beta V(\bfsigma (i),\bfsigma (i+1))\Bigr\}.
\end{eqnarray}
Here $n(\bft (i))$ stands for the number of triangles in the triangulation $\bft (i)$.
Like before, the formula
\begin{equation}\label{trmatrix}
\Xi_N(\mu,\beta )={\rm tr}\;\bfK^N
\end{equation}
gives rise to a transfer matrix $\bfK$ with entries
$K((\bft,\bfsigma),(\bft',\bfsigma'))$ labelled by pairs
$(\bft,\bfsigma),(\bft',\bfsigma')$ representing triangulations of
a single strip (say, ${\mathcal S}\times [0,1]$) and their supported spin
configurations which are positioned next to each other. Formally,
\begin{eqnarray}\label{yamb-tm}
K((\bft,\bfsigma),(\bft^\prime,\bfsigma^\prime))&=&{\mathbf 1}_{\bft \sim \bft^\prime}
\exp\Bigl\{ - \frac{\mu}{2} (n(\bft )+n(\bft')) \Bigr\} \\ \nonumber
&\times& \exp\Bigl\{ -\frac{\beta}{2}\bigl( H(\bfsigma)+H(\bfsigma') \bigr)
-\beta V(\bfsigma , \bfsigma')\Bigr\}.
\end{eqnarray}
As earlier, $n(\bft)$ and $n(\bft^\prime )$ are the numbers of triangles in
the triangulations
$\bft$ and $\bft^\prime$. The indicator ${\mathbf 1}_{\bft \sim \bft^\prime}$
means that the triangulations $\bft, \bft^\prime$ have to be consistent 
with each
other in the above sense: the number of down-triangles in $\bft$ should equal the number of
up-triangles in $\bft^\prime$, and an upper-marked edge in $\bft$ should coincide
with a lower-marked edge
in triangulation $\bft^\prime$. It means that the pair $(\bft,\bft')$ forms a CDT for 
the strip $\mathcal S \times [0,2]$.

We would like to stress that the trace ${\rm tr}\;\bfK^N$ in \eqref{trmatrix} is understood
as the {\it matrix trace}, i.e., as the sum $\sum_{\bft,\bfsigma}K^{(N)}(
(\bft,\bfsigma),(\bft,\bfsigma ))$ of the diagonal entries $K^{(N)}(
(\bft,\bfsigma),(\bft,\bfsigma ))$  of the matrix ${\bf K}^N$. (Indeed, in what follows,
the notation ``$\,{\rm{tr}}\,$'' is used for the matrix trace only.) Our aim will be to 
verify that the matrix trace in  \eqref{trmatrix} can be replaced with an {\it 
operator trace} invoking the eigenvalues of ${\bf K}$ in a suitable linear space. 

As before, we can introduce the $N$-strip Gibbs probability distribution
associated with formula (\ref{yamb-pf}):
\begin{eqnarray}\label{yamb-Gd}
&& \mathbb P_N \bigl( (\bft (0),\bfsigma (0)),\ldots ,(\bft (N-1),\bfsigma (N-1)) \bigr) \\ \nonumber && {}
=\frac{1}{\Xi (\mu,\beta )} \prod_{i=0}^{N-1}
\exp\Bigl\{ -\mu n(\bft (i)) -\beta H(\bfsigma (i))-\beta V(\bfsigma (i),\bfsigma (i+1))\Bigr\}.
\end{eqnarray}

Consider several special cases of interest.

\begin{description}

\item{\underline{The case $\beta \approx 0$}.}
This is the first term of the so-called high
temperature expansion \cite{Ambjorn:1999gi}. Here one has
\begin{eqnarray*}
  \Xi(\mu,0) &=& \sum_{(\bft (0), \dots, \bft (N-1)) }\exp \Bigl\{ -\mu \sum_{i=0}^{N-1} n(\bft(i)) \Bigr\}
\sum_{(\bfsigma (0),\dots ,\bfsigma (N-1))} 1 \\
   &=& \sum_{n^0\ge 1, \dots ,n^{N-1}\ge 1} \exp \Bigl\{-2(\mu -\ln 2) \sum_{i=0}^{N-1}n^i  \Bigr\} \prod_{i=0}^{N-1} \binom{n^i + n^{i+1}-1}{n^i -1} \\
   &=&Z_N(\mu -\ln 2);\;\hbox{ cf. (\ref{yamb-pf1}).} \phantom{\sum_{N}^N}
\end{eqnarray*}
The condition $\diy\mu -\ln 2 >\ln 2$ which guarantees properties listed in Lemma
\ref{sz-l3} and Theorem \ref{theo1} resuls in
\begin{equation}
\label{yamb-in2} \mu > 2\ln 2.
\end{equation}
Thus, Eqn.\ \eqref{yamb-in2} yields a sub-criticality condition when $\beta =0$.



\item{\underline{The case $\beta \approx \infty$}.} Observe that for any triangulation
$\ubft = (\bft(0), \dots, \bft(N-1))$ there are two ground states: all spins $+1$ and all spins $-1$, with the
overall energy equals minus three half times the total number of
triangles: $-3/2\, \sum_{i=0}^{N-1} n(\bft(i)).$ Discarding all other spin configurations, we obtain that
$$\Xi (\mu,\beta )>\Xi_*(\mu,\beta )
$$
where
\begin{eqnarray*}
  \Xi_*(\mu,\beta ) &=& \sum_{\bft (0),\dots ,\bft (N-1)}
2\exp\Bigl\{ \bigl(-\mu +\frac{3}{2}\beta \bigr) \sum_{i=0}^{N-1} n(\bft(i)) \Bigr\} \\
  &=& 2 \sum_{n^0\ge 1, \dots , n^{N-1} \ge 1 }
\exp \Bigl\{ -2\bigl(\mu- \frac{3}{2}\beta \bigr) \sum_{i=0}^{N-1} n^i \Bigr\} \binom{n^i+n^{i+1}-1}{n^i -1} \\ &=&
2 Z_N\left(\mu-\frac{3}{2}\beta \right) \phantom{\sum_N^N}
\end{eqnarray*}
where $\displaystyle\exp\left[\frac{3}{2}\beta\sum_i n(\bft(i)) \right]$ is the energy of the $(+)$-configuration
(or, equivalently, the $(-)$-configuration). For $\beta$ large,  we can expect that
$\Xi (\mu,\beta )\sim\Xi_*(\mu,\beta )$.
Then the critical inequality
$$\mu - \frac{3}{2}\beta > \ln 2$$
yields
\begin{equation}\label{yamb-in1}
\mu > \ln 2 + \frac{3}{2}\beta.
\end{equation}
Equation \eqref{yamb-in1} gives a necessary (and probably tight) criticality condition for
the Ising model under consideration for large values of $\beta$.
A similar result was  obtained in \cite{Ambjorn:1999gi}.

\item{\underline{The case $0 < \beta < \infty$}.}
Firstly, we note that for any fixed triangulation $\ubft$ the
energy of any spin configuration $\ubfsigma$ on $\ubft$ will be bigger or equal than the energy of a pure configuration (all $+$s or all $-$s):
\begin{eqnarray*}
H(\ubfsigma) &=& \sum_{j} H(\bfsigma (j)) + \sum_j V(\bfsigma (j), \bfsigma (j+1)) \\
&\ge& - \frac{3}{2} \#(\mbox{of all triangles in }\underline{t}) = - 3 \sum_{i=0}^{N-1} n^i,
\end{eqnarray*}
where $n^i$ is the number of edges in the $i$th level $S\times\{i\}, i=0,1\dots, N-1$.
Thus, for any $\beta >0$ the inequality $\Xi (\mu,\beta )<\Xi^*(\mu,\beta )$ holds true, where
\begin{eqnarray*}
  \Xi^*(\mu,\beta) &=& \sum_{ (\bft(0),\dots, \bft(N-1) } \exp\Bigl\{ \bigl( -\mu+\frac{3}{2} \beta + \ln 2 \bigr)  \sum_{i=0}^{N-1} n(\bft(i)) \Bigr\}  \\
&=& \sum_{ n^0\ge 1,\dots, n^{N-1}\ge 1 } \exp\Bigl\{ -2\bigl( \mu-\frac{3}{2} \beta 
- \ln 2 \bigr)  \sum_{i=0}^{N-1} n^i \Bigr\}  \\
&=& Z_N \bigl(\mu - \frac{3}{2}\beta - \ln 2 \bigr) \phantom{\sum_N^N}.
\end{eqnarray*}
Hence, the inequality
\begin{equation}\label{suffco}
\mu - \frac{3}{2}\beta - \ln 2 >\ln 2\ \ \ \hbox{ or }\ \ \  \mu > 2\ln 2 + \frac{3}{2} \beta 
\end{equation}
provides a sufficient condition for subcriticality of the Ising model  under consideration.

\end{description}

\subsection{The transfer-matrix ${\bf K}$ and its powers ${\bf K}^N$}

The main results of this article are summarized in Lemma \ref{yamb-l2} and Theorems
\ref{theo2} and \ref{theo3} below. Let us start with a statement (see Proposition \ref{KK}
below)
which merely re-phrases standard definitions and explains our interest in the matrices 
${\bf K}$, ${\bf K}^{\rm T}$, ${\bf K}^{\rm T}{\bf K}$, ${\bf K}{\bf K}^{\rm T}$ and their powers.   
Cf. Definition 2.2.2 on p.83, Definition 2.4.1 on p.101, Lemma 2.3.1 on p.85 and Theorem 
3.3.13 on p.139 in \cite{Ringrose}). 

We treat 
the transfer-matrix $\bfK$ and its transpose $\bfK^{\rm T}$ as linear operators
in the Hilbert space $\ell^2_{\rm{T-C}}$ (the subscript T-C refers to triangulations 
and spin-configurations). The space $\ell^2_{\rm{T-C}}$ is
formed by functions $\bfpsi =\{\bfpsi (\bft,\bfsigma )\}$ with
the argument $(\bft,\bfsigma )$ running over single-strip triangulations and
supported configurations of spins, with the scalar product
$\left\langle\bfpsi',\bfpsi''\right\rangle_{\rm T-C}
=\sum_{\bft ,\bfsigma}\bfpsi' (\bft ,\bfsigma )
{\overline{\bfpsi''}}(\bft ,\bfsigma )$ and the induced norm $\|\bfpsi\|_{\rm T-C}$. 
 The action of $\bfK$ in $\ell^2_{\rm{T-C}}$, in the basis formed by 
Dirac's delta-vectors $\delta_{(\bft ,\bfsigma)}$, is determined by
\begin{equation}\label{action}
\big(\bfK\bfpsi\big)(\bft,\bfsigma ) =\sum_{\bft',\bfsigma'}
K((\bft ,\bfsigma),(\bft',\bfsigma'))
\bfpsi (\bft',\bfsigma');
\end{equation}
in following we use the notation $\bfK$, $\bfK^{\rm T}$, etc., for the matrices
and the corresponding operators in $\ell^2_{\rm{T-C}}$.  
Accordingly, the symbols $\|\bfK\|_{\rm{T-C}}$, $\|\bfK^{\rm T}\|_{\rm{T-C}}$
etc. refer to norms in $\ell^2_{\rm{T-C}}$. 

Given $n=1,2,\ldots$, suppose that the operator ${\bf K}^n$ (respectively,
 $\left({\bf K}^{\rm T}\right)^n$) is of trace class . Then the following series 
absolutely converges:
\begin{equation}\label{optr}\sum_j\bfLambda^{(n)}_j\;\;\left(\hbox{respectively,}\;\sum_j{\bfLambda^*}^{(n)}_j \right),
\end{equation}
where $\bfLambda^{(n)}_j$ ($ {\bfLambda^*}^{(n)}_j$) runs through
the eigenvalues of ${\bf K}^n$ ($({\bf K}^{\rm T})^n$), counted with their multiplicities.
In this case the sum \eqref{optr} is called the operator trace of ${\bf K}^n$ (respectively, 
$({\bf K}^{\rm T})^n$) in $\ell^2_{\rm{T-C}}$. We adopt an agreement that the eigenvalues
in \eqref{optr} are listed in the decreasing order of their moduli, beginning with $\bfLambda^{(n)}_0$
(${\bfLambda^*}^{(n)}_0$) .

Set $\left|{\bf K}^n\right|=\sqrt{\left({\bf K}^{\rm T}\right)^n{\bf K}^n}$ and 
$\left|\left({\bf K}^{\rm T}\right)^n\right|=\sqrt{{\bf K}^n\left({\bf K}^{\rm T}\right)^n}$.

\begin{prop}\label{KK} For any positive integer $r$, the following inequalities are
equivalent:
\begin{equation}\label{KK1}
\begin{array}{l}
{\rm tr}\left(\left({\bf K}^r({\bf K}^{\rm T}\right)^{r}\right)= {\rm tr}\left( 
\left({\bf K}^{\rm T}\right)^r {\bf K}^r\right) < \infty 
\ \mbox{ and }\\  {\rm tr}|{\bf K}^{2r}| ={\rm tr} |({\bf K}^{\rm T})^{2r}|< \infty. \end{array}
\end{equation}
Moreover, each of the inequalities in \eqref{KK1} implies that $\forall$ $N\geq 2r$, 
the operators ${\bf K}^N$ and $({\bf K}^{\rm T})^N$ are of trace class  in $\ell^2_{\rm{T-C}}$. 
Hence, for $N\geq 2r$, the  
matrix traces ${\rm tr}\left({\bf K}^N\right)$ and  ${\rm tr}(({\bf K}^{\rm T})^N)$ 
are finite and coincide with the corresponding operator traces in $\ell^2_{\rm{T-C}}$.
\end{prop}

\begin{theo}\label{theo2}
Suppose that the condition (\ref{KK1}) is satisfied with $r=1$. 
Then the following properties of transfer matrix $\bfK$ are fullfilled.  
\begin{enumerate}
\item The square $\bfK^2$ and its transpose $(\bfK^{\rm T})^2$
are trace-class operators in $\ell^2_{\rm{T-C}}$. 

\item $\bfK$ and $\bfK^{\rm T}$ have a common eigenvalue,
$\bfLambda=\bfLambda_0(\beta ,\mu )>0$
such that the norms $\|\bfK\|_{\rm T-C}=\|\bfK^{\rm T}\|_{\rm T-C}=
\bfLambda$. Furthermore,  $\bfK^2$ and $(\bfK^{\rm T})^2$ have the 
common eigenvalue $\bfLambda^2=\bfLambda^{(2)}_0={\bfLambda^*}^{(2)}_0$
such that the norms $\|\bfK^2\|_{\rm T-C}=\| ( \bfK^{\rm T} )^2\|_{\rm T-C}=
\bfLambda^2$ . 

\item $\bfLambda$ is a simple eigenvalue of $\bfK$ and $\bfK^{\rm T}$, i.e., the
corresponding eigenvectors  $\bfphi =\{\bfphi (\bft ,\bfsigma )\}$
and $\bfphi^{*} =\{\bfphi^{*} (\bft ,\bfsigma )\}$  are unique
up to multiplicative constants. Moreover, $\bfphi$ and $\bfphi^{\rm T}$
can be made
strictly positive: $\bfphi (\bft ,\bfsigma ),\bfphi^{\rm T}(\bft ,\bfsigma )
>0$ $\forall$ $(\bft ,\bfsigma )$. Furthermore, $\bfLambda$  is separated 
from the remaining singular values and the remaining eigenvalues of
$\bfK$ and $\bfK^{\rm T}$ by a positive gap. The same is true for
$\bfLambda^2$ and $\bfK^2$ and $\left(\bfK^{\rm T}\right)^2$.
\end{enumerate}
\end{theo}

\noindent
{\bf Proof of Theorem \ref{theo2}.} Because the entries $K((\bft ,\bfsigma),(\bft',\bfsigma'))$
are non-negative, the condition \eqref{KK1} with $r=1$ means that
\begin{equation}\label{yamb-e22}
\sum_{(\bft ,\bfsigma),(\bft', \bfsigma')}
K^2((\bft ,\bfsigma),(\bft',\bfsigma'))
< \infty ,
\end{equation}
that is,  $\bfK$ and $\bfK^{\rm T}$
are Hilbert-Schmidt
operators. It means that the operator $\bfK\bfK^{\rm T}$ has an orthonormal basis of eigenvectors and
the series of squares of its eigenvalues (counted with multiplicities) converges
and gives the trace ${\rm tr}_{\rm T-C}(\bfK\bfK^{\rm T}$). Consequently, the operators $\bfK$ and $\bfK^{\rm T}$ are
bounded (and even completely bounded) and $\bfK^2$ and $({\bfK^{\rm T}})^2$ are of trace class.
The latter fact means that the matrix trace of the operator $\bfK^2$ coincides with its operator trace
in $\ell^2_{\rm{T-C}}$, and the same is true of $({\bfK^{\rm T}})^2$. In addition, the operator $\bfK^2$ has the
property that its matrix entries $K^{(2)}((\bft ,\bfsigma),(\bft',\bfsigma'))$
are strictly positive:
\begin{equation}\label{positivity}
K^{(2)}((\bft ,\bfsigma),(\bft',\bfsigma'))=
\sum_{(\wt\bft,\wt\bfsigma )}K((\bft ,\bfsigma),(\wt\bft,\wt\bfsigma ))
K((\wt\bft ,\wt\bfsigma),(\bft',\bfsigma'))>0.\end{equation}
The Krein--Rutman theory (see
\cite{Krein-Rutman}, Proposition VII$'$) guarantees that both $\bfK$ and $\bfK^{\rm T}$
have a maximal eigenvalue
$\bfLambda$ that is positive and non-degenerate, or simple. That is, the
eigenvector $\bfphi$ of $\bfK$ and the eigenvector $\bfphi^*$ of $\bfK^{\rm T}$
corresponding with $\bfLambda$ are unique up to multiplication
by a constant, and all entries $\bfphi (\bft, \bfsigma )$ and $\bfphi^*(\bft, \bfsigma )$ are non-zero and
have the same sign. In other words, the entries $\bfphi (\bft, \bfsigma )$ and $\bfphi^*(\bft, \bfsigma )$ 
can be made positive. The spectral gaps 
are also consequences of the above properties. $\Box$
\vskip 1 truecm

Set:
\begin{equation}\label{lambda}
\lambda (\mu,\beta)= c^2\,(m^2+1)\,({\rm cosh}\,2\beta)
\left( 1+ \sqrt{1-\frac{1}{({\rm cosh}\,2\beta)^2}\frac{(m^2 -1)^2 }{(m^2 + 1)^2}}\right)
\end{equation}
where $c$ and $m$ are determined by
\begin{eqnarray}\label{ccccc}
c&=&\frac{\exp(\beta-\mu )}{e^{2\beta}(1-\exp(\beta-\mu ))^2 - e^{-2\mu}} \\
\label{mmmmmm}
m &=& e^{2\beta} + (1-e^{4\beta})\exp\,(-(\beta +\mu) ).
\end{eqnarray}

\begin{lm}\label{yamb-l2}
For any $\beta, \mu >0$ such that
\begin{equation}\label{yamb-e14}
\lambda (\mu,\beta)<1,
\end{equation}
the condition \eqref{KK1} is satisfied for $r=1$:
\begin{equation}\label{KK3}
{\rm tr}({\bf K} {\bf K}^{\rm T})= {\rm tr}( {\bf K}^{\rm T}{\bf K}) < \infty
\ \mbox{ and }\ \ {\rm tr}|{\bf K}^2| ={\rm tr}|({\bf K}^{\rm T}
)^2|< \infty ,
\end{equation}
implying the assertions of Proposition \ref{KK} and Theorem \ref{theo2}. Moreover, the 
condition \eqref{KK1} implies \eqref{yamb-e14}
\end{lm}
\noindent The proof of Lemma \ref{yamb-l2} is given in the next section. Here we only remark that the proof
is based on the following representation of the trace (\ref{KK3}): there exists matrices $Q, Q_m$ (see formulas \eqref{q} and \eqref{qm}) with positive entries and of size $4\times 4$ such that (see formula \eqref{yamb-e8}) 
\begin{equation*}
{\rm tr}( {\bf K} {\bf K}^{\rm T}) = {\rm tr} \Bigl( \Bigl( \sum_{k\ge 1}  Q^k \Bigr) Q_m \Bigr) + \ldots.
\end{equation*}
The convergence of the matrix series $\sum_{k\ge 1}  Q^k$ is equivalent to the condition that the maximal eigenvalue of the matrix $Q$ is less then 1. This is exactly the condition (\ref{yamb-e14}).
\vskip 1 truecm

\begin{theo}\label{theo3}
Under condition (\ref{yamb-e14}), the following limit holds:
\begin{equation}\label{yamb-e13A}
\lim_{N\to\infty}\frac{1}{N}\log\,\Xi_N(\beta ,\mu )=\log\,\bfLambda.
\end{equation}
Moreover, as $N\to\infty$, the $N$-strip Gibbs measure ${\mathbb P}_N$ (see Eqn (\ref{yamb-Gd})) converges
weakly to a limiting probability distribution
${\mathbb P}$ that is represented by a positive recurrent Markov chain with states
$(\bft ,\bfsigma )$, the transition matrix\\ $\bfP =\{P ((\bft,\bfsigma ),(\bft',\bfsigma'))\}$
and the invariant distribution $\bfpi =\{\pi(\bft,\bfsigma )\}$ where
\begin{eqnarray*}
P((\bft,\bfsigma ),(\bft',\bfsigma')) &=& \frac{K((\bft,\bfsigma ),(\bft',\bfsigma'))\bfphi (\bft',\bfsigma')}{\bfLambda\bfphi(\bft,\bfsigma)} \\
\pi (\bft,\bfsigma ) &=& \bfphi(\bft,\bfsigma )\bfphi^{\rm T}(\bft,\bfsigma )
\left/\big\langle \bfphi ,\bfphi^{\rm T}
\big\rangle_{\rm T-C}\right.
\end{eqnarray*}
with the norm
$\big\|\bfphi\big\|^2_{\rm T-C}=\sum_{\bft,\bfsigma}\bfphi(\bft,\bfsigma )^2.
$
\end{theo}

\noindent
{\bf Proof of Theorem \ref{theo3}.} The spectral gap for $\bfK$ implies that
$\forall$ $\bfpsi\in\ell^2_{\rm T-C}$, we have the convergence
$$\lim_{N\to\infty}\frac{1}{\bfLambda^N}\bfK^N\bfpsi =\left(\langle\bfpsi ,\bfphi\rangle_{\rm T-C}
\right)\bfphi$$
in the norm of space $\ell^2_{\rm T-C}$. Moreover, let $\bfPi$ denote the
operator of projection to the subspace spanned by the eigenvectors of
$\bfK$ different from $\bfphi$. Then
$$\diy\frac{1}{\bfLambda}\|\bfPi\bfK\bfP\|_{\rm T-C}<1 \ \ \Longrightarrow \ \ \diy\lim_{N\to\infty}\frac{1}{\bfLambda^N}\left\|
\big(\bfPi\bfK\bfP\big)^N\right\|_{\rm T-C}=0.$$
In turn, this implies that
$$\frac{1}{N}\log\;\Xi_N(\mu,\beta )=\frac{1}{N}\log\;{\rm tr}_{\rm T-C}\bfK^N
\to\log\bfLambda .$$
Convergence of the Gibbs measure ${\mathbb P}_N$ follows as a corollary.
 $\Box$
\vskip 1 truecm

\begin{figure}
\begin{center}
\includegraphics[width=11cm]{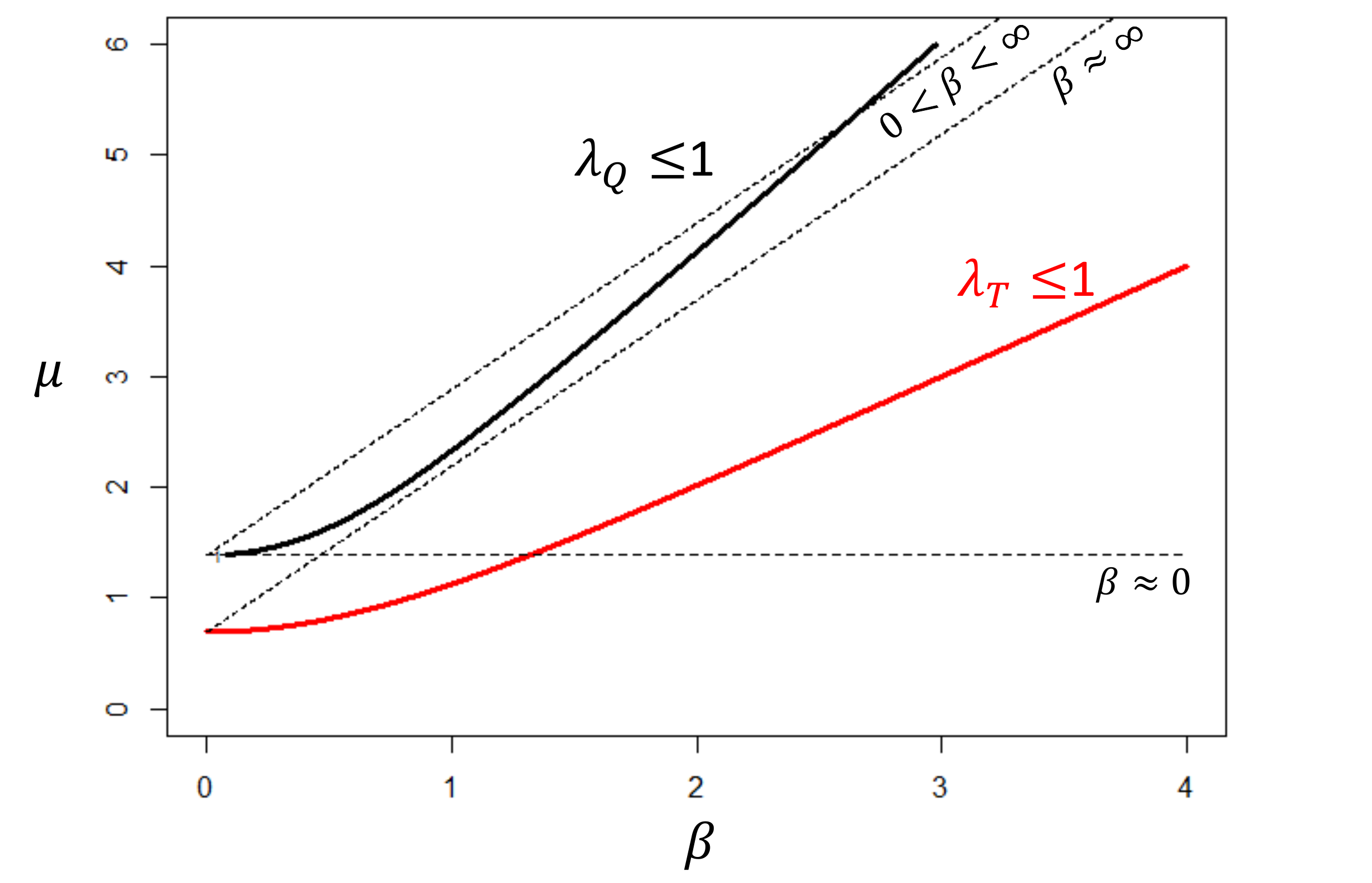}
 \end{center}
  \caption{$\lambda_Q=\lambda$ and $\lambda_T$ are the maximal eigenvalues of 
the matrix $Q$ and a related matrix $T$ respectively (see Appendix \ref{sec:app}). The area 
above the black curve is where the condition
 (\ref{yamb-e14}) holds true.}
  \label{fig0}
\end{figure}



\section{Concluding remarks} \label{sec5}

This paper makes a step towards determining the subcriticality
domain for an Ising-type model coupled to two-dimensional causal dynamical triangulations (CDT). In doing so we employ transfer-matrix techniques and in particular the Krein-Rutman theorem. We complement the discussion of the previous sections with the following two concluding remarks:

\begin{remark}
It is instructive to summarise the logical structure
of the argument establishing Lemma \ref{yamb-l2} and Theorems \ref{theo2}
and \ref{theo3}:
\begin{itemize}
\item First, \eqref{KK3} holds iff condition \eqref{yamb-e14} holds: see
the proof of Lemma \ref{yamb-l2}. 
\item Next, \eqref{KK3} implies that $\bfK$ is a Hilbert--Schmidt operator 
and $\bfK^2$ is a trace class operator in
$\ell^2_{\rm T-C}$. 
\item 
The last fact,
together with the property of positivity \eqref{positivity}, allow us to use
the Krein--Rutman theory, deriving all conclusions of Theorems
\ref{theo2} and \ref{theo3}.
\end{itemize}

On the other hand, if \eqref{yamb-e14} fails (and therefore
\eqref{KK3} fails), it does
not necessarily mean that the assertions
Theorems \ref{theo2} and \ref{theo3} fail. In other words, we do not claim that
the boundary of the domain of parameters $\beta$ and $\mu$ where the model
exhibits subcritical behavior is given by  Eqn.\
\eqref{yamb-e14}.  Moreover, Figure \ref{fig0} shows the result of a numerical calculation  
indicating that the condition \eqref{yamb-e14} 
is worse than \eqref{suffco} for (moderately) large values of $\beta$.

An apparent condition closer to necessity is the pair of inequalities
\eqref{KK1} for some (possibly) large $r$. This issue needs a further study.

\end{remark}

\begin{remark}
Physical considerations suggest that
the critical curve in the $(\beta, g)$ quarter-plane would have some predictable
patterns of behavior: as a function of $\beta$, it would decay and exhibit a first-order
singularity at a unique point $\beta =\beta_{\rm{cr}}\in (0,\infty )$.

A plausible conjecture is that the boundary of the critical domain coincides with
the locus of points $(\beta ,\mu )$ where $\bf\Lambda$ looses either
the property of positivity or the property of being a simple eigenvalue. This direction 
also requires further research.
\end{remark}


\subsection*{Acknowledgements.} This work was supported by FAPESP 2012/04372-7.  
JCH acknowledges support by CAPES. YS would like to thank the FAPESP foundation for the financial support and NUMEC, IME University of S\~ao Paulo, for
warm hospitality. The work of AY was partially supported CNPq 308510/2010-0. The work of SZ was partially supported by
FAPERJ 111.859/2012, CNPq 307700/2012-7 and PUC-Rio. Further, he thanks the IME at the University of S\~ao Paulo, as well as the Rudolf Peierls Centre for Theoretical Physics and Mansfield College, University of Oxford for kind hospitality and financial support during visits.

\appendix

\section{ Proof of Lemma \ref{yamb-l2}.} \label{sec:app}
By definition the trace \eqref{KK3} we need to calculate the series
\begin{eqnarray}\label{ntr}
{\rm tr}({\bf K}^{\rm T} {\bf K}) &=& \sum_{(\bft,\bfsigma )} {\bf K}^{\rm T} {\bf K} ((\bft,\bfsigma ),(\bft,\bfsigma)) \nonumber \\
&=& \sum_{(\bft,\bfsigma ), (\bft',\bfsigma')} K ((\bft,\bfsigma ), (\bft',\bfsigma')) K ((\bft,\bfsigma ), (\bft',\bfsigma')) \nonumber \\
&=& \sum_{(\bft,\bfsigma ), (\bft',\bfsigma')} K^2 ((\bft,\bfsigma ), (\bft',\bfsigma')).
\end{eqnarray}

A single-strip triangulation $\bft$
consists of up- and down-triangles. Accordingly, it is convenient
to employ new labels for spins: if a triangle $t(l)$ is an $l$th up-triangle
then we denote it by $t^l_{\rm up}$; the corresponding spin $\sigma (j)$
will be denoted by $\sigma^l_{\rm up}$. Similarly,
if $t(j)$ is an $l$th down-triangle then we denote it by $t^l_{\rm do}$;
the spin $\sigma (j)$ will be
denoted by $\sigma^l_{\rm do}$. Consequently, the triangulation $\bft$
and its supported spin-configuration $\bfsigma$ are
represented as
$$\bft :=(\bft_{\rm up},\bft_{\rm do})\;\hbox{ and }\;
\bfsigma :=(\bfsigma_{\rm up},\bfsigma_{\rm do}).$$
Here
$$\bft_{\rm up}=(t^1_{\rm up},\ldots, t^n_{\rm up}),\;
\;\bft_{\rm do}=(t^1_{\rm do},\ldots ,t^m_{\rm do}),$$
and
$$\bfsigma_{\rm up}=(\sigma^1_{\rm up},\ldots, \sigma^n_{\rm up}),\;
\;\bfsigma_{\rm do}=(\sigma^1_{\rm do},\ldots ,\sigma^m_{\rm do}),$$
assuming that the supporting single-strip triangulation $\bft$ contains $n$
up-triangles and $m$ down-triangles. (The actual order of up- and down-triangles
and supported spins does not matter.)

The same can be done for the pair $(\bft',\bfsigma')$ as illustrated in
(\ref{ntr}).
Let recall that the triangulations $\bft$ and $\bft'$ are consistent ($\bft\sim\bft'$)
iff number of the down-triangles in $\bft$ equals that of up-triangles
in $\bft'$.

\begin{figure}
\begin{center}
\includegraphics[width=11cm]{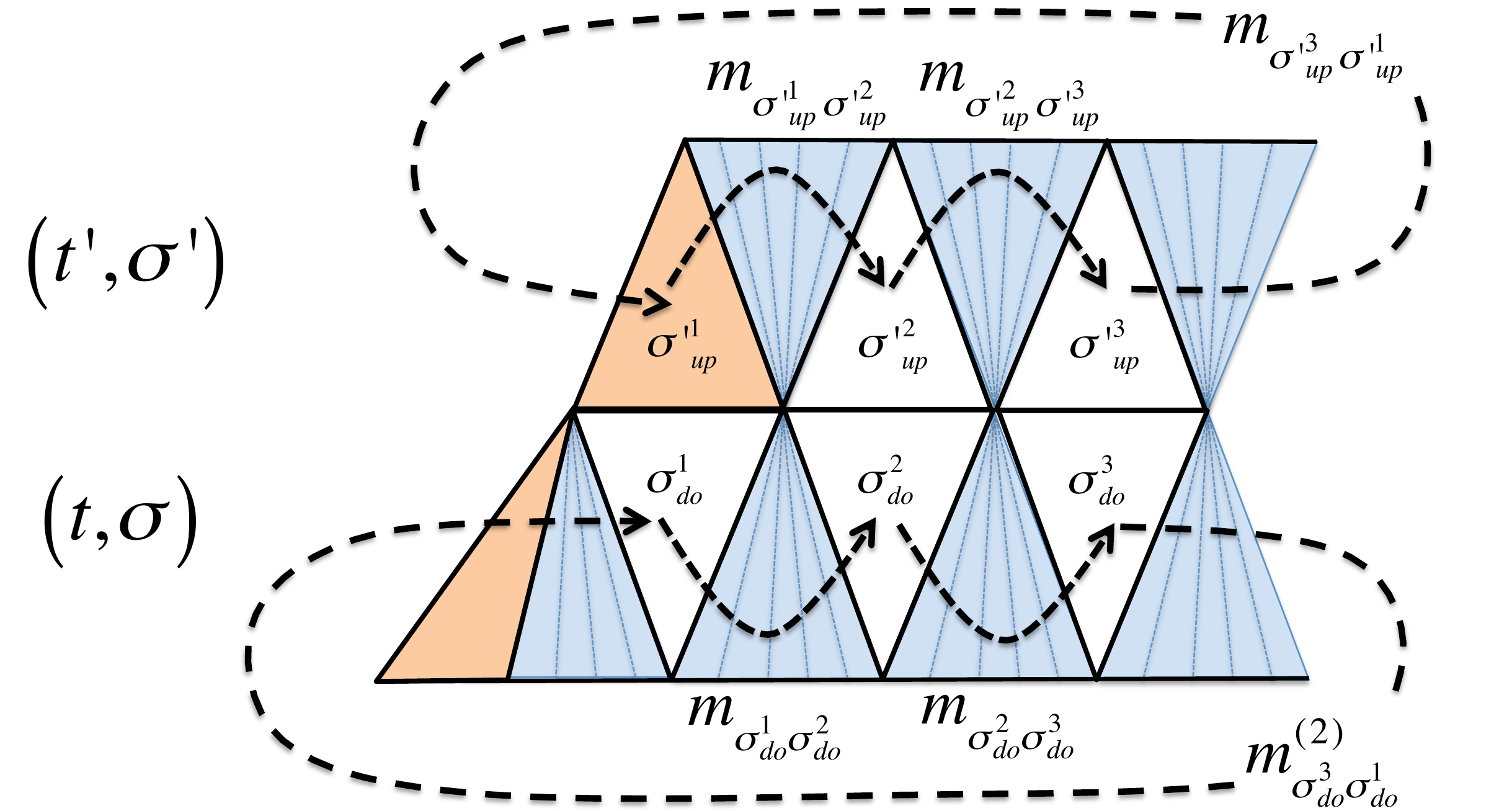}
   \end{center}
  \caption{ Interaction  betwen the elements  $z_i$, $\eta_i$ }
  \label{fig14}
\end{figure}

To calculate the sum (\ref{ntr}) we divide the summation over
$(\bft',\bfsigma')$ into a summation over $(\bft'_{\rm up},\bfsigma'_{\rm up})$
and $(\bft'_{\rm do},\bfsigma'_{\rm do})$. Firstly, fix a pair
$(\bft'_{\rm up},\bfsigma'_{\rm up})$ and make the sum over $(\bft'_{\rm do},
\bfsigma'_{\rm do})$.  Note that the term $V((\bft,\bfsigma),(\bft',\bfsigma'))$
depends only on $\bfsigma_{\rm do}$ and $\bfsigma'_{\rm up}$. Consequently,

\begin{eqnarray}\label{yamb-e3A}
&& \sum_{\bft'_{\rm do},\bfsigma'_{\rm do}}K^2 ((\bft,\bfsigma),(\bft',\bfsigma'))\\
&& {} = e^{-\beta H(\bfsigma)} e^{-2\beta V((\bft,\bfsigma ),(\bft',\bfsigma')) }
e^{-\mu n(\bft)} \sum_{(\bft'_{\rm do},\bfsigma'_{\rm do})}
e^{-\beta H(\bfsigma')} e^{-\mu n(\bft')}. \nonumber
\end{eqnarray}

\noindent
The sum in the right-hand side of (\ref{yamb-e3A}) can be represented
in a matrix form. Denote by
$e_{\pm 1}$ the standard spin-$1/2$ unit vectors in ${\mathbb R}^2$:
$$e_{+1}=\begin{pmatrix}1\\ 0\end{pmatrix}\;\;\hbox{ and }\;\;
e_{-1}=\begin{pmatrix}0\\ 1\end{pmatrix}\,.$$
Next, let us introduce a $2\times 2$ matrix $T$ where
\begin{equation}\label{yamb-e4}
T = e^{-\mu}\begin{pmatrix} e^{\beta} & e^{-\beta} \\
e^{-\beta} & e^{\beta}\end{pmatrix} := \begin{pmatrix} t_{++} & t_{+-} \\
t_{-+} & t_{--}\end{pmatrix}.
\end{equation}
Denote by $n(i), i=1,\dots, n_{up}(\bft' )$ the number of down-triangles in $\bft'$ which are between the $i$th and $(i+1)$th up-triangles in $\bft'$. Let $n_{up}(\bft' )=k$ then
\begin{eqnarray}\label{yamb-e3.1}
\nonumber 
&& \sum\limits_{\bft'_{\rm do},\bfsigma'_{\rm do}}e^{-\beta H(\bfsigma')}
e^{-\mu n(\bft')} = \sum_{n(i)\ge 0:\  \sum_i n(i) \ge 1}
\prod\limits_{l=1}^k
\left(e^{\rm T}_{{\sigma'}^l_{\rm up}} T^{n(l)+1}
e_{{\sigma'}^{l+1}_{\rm up}}\right)\\
&& {} = \prod\limits_{l=1}^k\left(e^{\rm T}_{{\sigma'}^l_{\rm up}} M
e_{{\sigma'}^{l+1}_{\rm up}}\right) - \prod_{l=1}^k  \left(e^{\rm T}_{{\sigma'}^l_{\rm up}} T e_{{\sigma'}^{l+1}_{\rm up}}\right)
\end{eqnarray}
where the matrix $M$ is the sum of the geometric progression
\begin{equation}\label{yamb-e6}
M = \sum_{n=1}^{\infty} T^n
:=\begin{pmatrix}m_{++} & m_{+-} \\ m_{-+} & m_{--} \end{pmatrix}\,.
\end{equation}
Using the same procedure we can obtain the sum over all up-triangles into the triangulation $\bft$. The only difference is the existence of marked up-triangle in the strip: let as before $n_{up}(\bft' )=n_{do}(\bft)=k$ then
\begin{equation}\label{yamb-e3.2}
\sum\limits_{\bft_{\rm up},\bfsigma_{\rm up}}e^{-\beta H(\bfsigma)}
e^{-\mu n(\bft)} = \prod\limits_{l=1}^{k-1} \left(e^{\rm T}_{{\sigma}^l_{\rm up}} M
e_{{\sigma}^{l+1}_{\rm up}}\right) \left(e^{\rm T}_{{\sigma}^{k}_{\rm up}} M^2 e_{{\sigma}^{1}_{\rm up}}\right)
\end{equation}
See Figure~\ref{fig14} for illustration of these calculations \eqref{yamb-e3.1} and \eqref{yamb-e3.2}. Further,  supposing the existence of the matrix $M$ and using (\ref{yamb-e3.1}) and (\ref{yamb-e3.2}) we obtain the following:
\begin{eqnarray}
\nonumber && \sum_{\bft_{\rm up},\bfsigma_{\rm up}} \sum_{\bft'_{\rm do},\bfsigma'_{\rm do}} K^2((\bft,\bfsigma),(\bft',\bfsigma')) = e^{-2\beta V((\bft_{\rm do},\bfsigma_{\rm do}),  (\bft'_{\rm up},\bfsigma'_{\rm up}) ) }  \\ \nonumber
&& \ \ \  \times \sum_{\bft_{\rm up},\bfsigma_{\rm up}} e^{-\beta H(\bfsigma)} 
e^{-\mu n(\bft)} \sum_{(\bft'_{\rm do},\bfsigma'_{\rm do})}
e^{-\beta H(\bfsigma')} e^{-\mu n(\bft')} \\ 
&& {} = e^{-2\beta V((\bft_{\rm do},\bfsigma_{\rm do}),  (\bft'_{\rm up},\bfsigma'_{\rm up}) ) }  \nonumber \\ && {} \times \Bigl[ \prod\limits_{l=1}^k\left(e^{\rm T}_{{\sigma'}^l_{\rm up}} M
e_{{\sigma'}^{l+1}_{\rm up}}\right)  \prod\limits_{l=1}^{k-1} \left(e^{\rm T}_{{\sigma}^l_{\rm do}} M e_{{\sigma}^{l+1}_{\rm do}}\right) \left(e^{\rm T}_{{\sigma}^{k}_{\rm up}} M^2 e_{{\sigma}^{1}_{\rm up}}\right) \nonumber \\
&& \ \ \ {} - \prod\limits_{l=1}^k\left(e^{\rm T}_{{\sigma'}^l_{\rm up}} T
e_{{\sigma'}^{l+1}_{\rm up}}\right) \prod_{l=1}^{k-1}  \left(e^{\rm T}_{{\sigma}^l_{\rm do}} M e_{{\sigma}^{l+1}_{\rm do}}\right) \left(e^{\rm T}_{{\sigma}^{k}_{\rm up}} M^2 e_{{\sigma}^{1}_{\rm up}}\right)  \Bigr] .\label{yamb-e3B}
\end{eqnarray}

A necessary and sufficient condition for the convergence of the matrix series
for $M$ is that the maximal eigenvalue of matrix $T$ is less then 1.
The eigenvalues of $T$ are
\begin{equation}\label{yamb-eigen1}
\lambda_\pm = e^{(\beta -\mu )} \pm e^{-(\beta +\mu )},
\end{equation}
and the above condition means that $\lambda_+<1$
or, equivalently,
\begin{equation}\label{yamb-cond1}
\mu > \ln\bigl( 2{\rm cosh}(\beta)\bigr).
\end{equation}
Under this condition \eqref{yamb-cond1}, the matrix $M$ is calculated explicitly:
\begin{equation}\label{yamb-e5}
\begin{array}{l}
\diy  M=\frac{e^{(\beta-\mu )}}{e^{2\beta}(1-e^{(\beta-\mu)})^2- e^{-2\mu} }\\
\;\\
\diy\qquad\times\begin{pmatrix}e^{2\beta}+(1-e^{4\beta})e^{-(\beta+\mu )}&1\\
1&e^{2\beta}+(1-e^{4\beta})e^{-(\beta+\mu )}\end{pmatrix}\,.\end{array}
\end{equation}

We are now in a position to calculate the sum in \eqref{ntr}.
To this end,  we again represent it through the product of transfer matrices. Pictorially, we
express the above sum as the partition function of a one-dimensional Ising-type
model where
states are pairs of spins $(\sigma^l_{\rm do}, \sigma^l_{\rm up})$ and the
interaction is via the matrix $T$ between the members of the pair and via matrix $M$
between neighboring pairs. More precisely, define the following $4\times 4$
matrices:
\begin{eqnarray}\label{q}
\!\!\!\!\!\!\!\! Q \!\!\!\!\!&=&\!\!\! \!\! \begin{pmatrix} e^{2\beta} m_{++}m_{++} & m_{++}m_{+-} & m_{+-}m_{++} & e^{2\beta}m_{+-}m_{+-} \\ m_{++}m_{-+} & e^{-2\beta}m_{++}m_{--} & e^{-2\beta}m_{+-}m_{-+} & m_{+-}m_{--} \\ m_{-+}m_{++} & e^{-2\beta}m_{-+}m_{+-} & e^{-2\beta}m_{--}m_{++} & m_{--}m_{+-} \\ e^{2\beta}m_{-+}m_{-+} & m_{-+}m_{--} & m_{--}m_{-+} & e^{2\beta}m_{--}m_{--} \end{pmatrix} \\
\label{qm}
\!\!\!\!\!\!\!\! Q_m\!\!\!\!\! &\!\!\!=\!\!\!& \!\!\! \!\! \begin{pmatrix} e^{2\beta} m_{++}m_{++}^{(2)} & m_{++}m_{+-}^{(2)} & m_{+-}m_{++}^{(2)} & e^{2\beta}m_{+-}m_{+-}^{(2)} \\ m_{++}m_{-+}^{(2)}  & e^{-2\beta}m_{++}m_{++}^{(2)}  & e^{-2\beta}m_{+-}m_{-+}^{(2)} & m_{+-}m_{--}^{(2)}\\ m_{-+}m_{++}^{(2)} & e^{-2\beta}m_{-+}m_{+-}^{(2)} & e^{-2\beta}m_{--}m_{++}^{(2)} & m_{++}m_{++}^{(2)}\\ e^{2\beta}m_{-+}m_{-+}^{(2)} & m_{-+}m_{--}^{(2)} & m_{--}m_{-+}^{(2)} & e^{2\beta}m_{--}m_{--}^{(2)} \end{pmatrix} \\
\label{qt}
\!\!\!\!\!\!\!\! Q_t\!\!\!\!\! &\!\!\!=\!\!\!&\!\!\!\!  \!\!  \begin{pmatrix} e^{2\beta} t_{++}m_{++} & t_{++}m_{+-} & t_{+-}m_{++} & e^{2\beta}t_{+-}m_{+-} \\ t_{++}m_{-+} & e^{-2\beta}t_{++}m_{--} & e^{-2\beta}t_{+-}m_{-+} & t_{+-}m_{--} \\ t_{-+}m_{++} & e^{-2\beta}t_{-+}m_{+-} & e^{-2\beta}t_{--}m_{++} & t_{--}m_{+-} \\ e^{2\beta}t_{-+}m_{-+} & t_{-+}m_{--} & t_{--}m_{-+} & e^{2\beta}t_{--}m_{--} \end{pmatrix} \\
\label{qtm}
\!\!\!\!\!\!\!\! Q_{tm}\!\!\!\! &\!\!\!=\!\!\!&\!\!\!\!  \!\! \begin{pmatrix} e^{2\beta} t_{++}m_{++}^{(2)} & t_{++}m_{+-}^{(2)} & t_{+-}m_{++}^{(2)} & e^{2\beta}t_{+-}m_{+-}^{(2)} \\ t_{++}m_{-+}^{(2)}  & e^{-2\beta}t_{++}m_{++}^{(2)}  & e^{-2\beta}t_{+-}m_{-+}^{(2)} & t_{+-}m_{--}^{(2)}\\ t_{-+}m_{++}^{(2)} & e^{-2\beta}t_{-+}m_{+-}^{(2)} & e^{-2\beta}t_{--}m_{++}^{(2)} & t_{++}m_{++}^{(2)}\\ e^{2\beta}t_{-+}m_{-+}^{(2)} & t_{-+}m_{--}^{(2)} & t_{--}m_{-+}^{(2)} & e^{2\beta}t_{--}m_{--}^{(2)} \end{pmatrix}
\end{eqnarray}
where $m_{ij}, m^{(2)}_{ij}$ and $t_{i,j} \ (i,j \in \{-, +\})$ are elements of the matrices $M, M^2,$ and $T$ respectively.

Now for the sum under consideration (\ref{ntr}) we obtain using representation (\ref{yamb-e3B})
\begin{eqnarray}\label{yamb-e8}
    && \nonumber \sum_{(\bft, \bfsigma), (\bft',\bfsigma')} K^2((\bft,\bfsigma),(\bft',\bfsigma'))  = 
\sum_{(\bft_{\rm do},\bfsigma_{\rm do}),  (\bft'_{\rm up},\bfsigma'_{\rm up})}  e^{-2\beta V((\bft_{\rm do},\bfsigma_{\rm do}),  (\bft'_{\rm up},\bfsigma'_{\rm up}) ) }   \\
&& \ \ \ {} \times \left[ \prod\limits_{l=1}^k\left(e^{\rm T}_{{\sigma'}^l_{\rm up}} M
e_{{\sigma'}^{l+1}_{\rm up}}\right)  \prod\limits_{l=1}^{k-1} \left(e^{\rm T}_{{\sigma}^l_{\rm do}} M e_{{\sigma}^{l+1}_{\rm do}}\right) \left(e^{\rm T}_{{\sigma}^{k}_{\rm up}} M^2 e_{{\sigma}^{1}_{\rm up}}\right)  \right. \nonumber \\
&& \ \ \ \ \ \ \left. {} - \prod\limits_{l=1}^k\left(e^{\rm T}_{{\sigma'}^l_{\rm up}} T
e_{{\sigma'}^{l+1}_{\rm up}}\right) \prod_{l=1}^{k-1}  \left(e^{\rm T}_{{\sigma}^l_{\rm do}} M e_{{\sigma}^{l+1}_{\rm do}}\right) \left(e^{\rm T}_{{\sigma}^{k}_{\rm up}} M^2 e_{{\sigma}^{1}_{\rm up}}\right)  \right]  \nonumber \\
&& {} = 
    \mbox{tr} \Bigl( \Bigl(\sum_{k=0}^\infty Q^k \Bigr) Q_m \Bigr) - \mbox{tr} \Bigl( \Bigl( \sum_{k=1}^\infty Q_t^k  \Bigr) Q_{tm} \Bigr).
\end{eqnarray}
By the construction the matrix $Q$ is greater then $Q_t$ elementwise. Thus the eigenvalue of matrix $Q$ is greater than the eigenvalue of the matrix $Q_t$ (it follows from the Perron-Frobenius theorem).  Therefore the necessary and sufficient condition for the convergence in (\ref{ntr}) is
that the largest eigenvalue of $Q$ is less than 1. It is possible to calculate its eigenvalue analytically. In order to express
the eigenvalues of $Q$ it is convenient to use notations \eqref{ccccc} and \eqref{mmmmmm}. In this notations the matrix $M$, i.e.\ \eqref{yamb-e5}, is represented as follows
$$
M=c\left( \begin{array}{cc} m & 1 \\ 1 & m \end{array} \right).
$$
The equations for the eigenvalues of $Q$ are:
\begin{eqnarray*}
\lambda_1 &=& c^2 e^\beta (m^2 -1) \\
\lambda_2 &=& c^2 e^{-\beta} (m^2 -1) \\
\lambda_3 &=& c^2 (m^2+1)({\rm cosh}\,\beta )
\left( 1- \sqrt{1-\frac{(m^2 -1)^2}{({\rm cosh}\,\beta)^2(m^2 + 1)^2}}\right) \\
\lambda_4 &=& c^2 (m^2+1)({\rm cosh}\,\beta )
\left( 1+ \sqrt{1-\frac{(m^2 -1)^2}{({\rm cosh}\,\beta)^2(m^2 + 1)^2}}\right)
\end{eqnarray*}
A straightforward inspection confirms that the largest eigenvalue
is given by $\lambda_4$. The condition $\lambda_4<1$ coincides with
\eqref{ntr}. This completes the proof of Lemma \ref{yamb-l2}.
$\Box$

\providecommand{\href}[2]{#2}\begingroup\raggedright\endgroup

\end{document}